%% file: main.tex
\documentclass[letterpaper,11pt]{article}
\usepackage[utf8]{inputenc}

\usepackage{amsmath, amsthm, amssymb}
\usepackage{algorithmicx}
\usepackage[table,xcdraw]{xcolor}
\usepackage[ruled,vlined,linesnumbered]{algorithm2e}
\usepackage{mathtools}
\usepackage[numbers]{natbib}
\usepackage{comment} 
\usepackage{tcolorbox} 
\usepackage{xfrac}
\usepackage{hyperref}
\usepackage{multirow}
\usepackage{caption}
\usepackage{bm}
\usepackage{newfloat}
\usepackage{enumitem}

\usepackage[margin=1in]{geometry}

\allowdisplaybreaks

\definecolor{mygreen}{RGB}{20,160,60}
\definecolor{myred}{RGB}{150,0,0}

\hypersetup{
     colorlinks=true,
     citecolor= mygreen,
     linkcolor= myred
}

\input{macro.tex}
\input{utils.tex}

\makeatletter
\renewcommand{\paragraph}{%
  \@startsection{paragraph}{4}%
  {\z@}{10pt}{-1em}%
  {\normalfont\normalsize\bfseries}%
}
\makeatother

\makeatletter
\patchcmd{\@algocf@start}
  {-1.5em}
  {0pt}
  {}{}
\makeatother

\title{Improved Analysis of EDCS\\ via Gallai-Edmonds Decomposition}

\author{Soheil Behnezhad}

\date{}

\begin{document}

\maketitle

\input{abstract}
\thispagestyle{empty}
\setcounter{page}{0}

\input{intro.tex}

\input{thealgorithm}

\bibliographystyle{alpha}
\bibliography{references}
	
\end{document}

%% file: macro.tex
\renewcommand{\epsilon}{\varepsilon}
\newcommand{\stackeq}[2]{\ensuremath{\stackrel{\text{#1}}{#2}}}

\newcommand{\hiddencomment}[1]{}

\newcommand{\mc}[1]{\ensuremath{\mathcal{#1}}}

%% file: utils.tex
\usepackage[noabbrev,nameinlink]{cleveref}
\crefname{claim}{claim}{claims}

\newtheorem{theorem}{Theorem}[section]
\newtheorem{lemma}[theorem]{Lemma}
\newtheorem{proposition}[theorem]{Proposition}
\newtheorem{corollary}[theorem]{Corollary}

\newtheorem{definition}[theorem]{Definition}
\newtheorem{claim}[theorem]{Claim}
\newtheorem{fact}[theorem]{Fact}
\newtheorem{observation}[theorem]{Observation}

\definecolor{mylightgray}{RGB}{230,230,230}

\algnewcommand{\IIf}[1]{\State\algorithmicif\ #1\ \algorithmicthen}
\algnewcommand{\EndIIf}{\unskip\ \algorithmicend\ \algorithmicif}

\newenvironment{graytbox}{
\par\addvspace{0.1cm}
\begin{tcolorbox}[width=\textwidth,
                  boxsep=5pt,
                  left=1pt,
                  right=1pt,
                  top=2pt,
                  bottom=2pt,
                  boxrule=0pt,
                  arc=0pt,
                  colback=mylightgray,
                  colframe=black,
                  ]
}{
\end{tcolorbox}
}

\newenvironment{highlighttechnical}{
\par\addvspace{0.1cm}
\begin{tcolorbox}[width=\textwidth,
                  boxsep=5pt,
                  left=1pt,
                  right=1pt,
                  top=2pt,
                  bottom=2pt,
                  boxrule=1pt,
                  arc=0pt,
                  colframe=black,
                  colback=white
                  ]
}{
\end{tcolorbox}
}

\newenvironment{myproof}{
\vspace{-0.5cm}
\begin{proof}
}{
\end{proof}
}

%% file: abstract.tex
\begin{abstract}
In this note, we revisit the {\em edge-degree constrained subgraph} (EDCS) introduced by Bernstein and Stein (ICALP'15). An EDCS is a sparse subgraph satisfying simple edge-degree constraints that is guaranteed to include an (almost) $\frac{2}{3}$-approximate matching of the base graph. Since its introduction, the EDCS has been successfully applied to numerous models of computation. Motivated by this success, we revisit EDCS and present an improved bound for its key property in {\em general} graphs. 

\medskip
Our main result is a new proof of the approximation guarantee of EDCS that builds on the graph's Gallai-Edmonds decomposition, avoiding the probabilistic method of the previous proofs. As a result, we get that to obtain a $(\frac{2}{3} - \epsilon)$-approximation, a sparse EDCS with maximum degree bounded by $O(1/\epsilon)$ is sufficient. This improves the $O(\log(1/\epsilon)/\epsilon^2)$ bound of Assadi and Bernstein (SOSA'19) and the $O(1/\epsilon^3)$ bound of Bernstein and Stein (SODA'16). Our guarantee essentially matches what was previously only known for bipartite graphs, thereby removing the gap in our understanding of EDCS in general vs. bipartite graphs.
\end{abstract}

%% file: intro.tex
\clearpage
\section{Introduction}

\newcommand{\EDCS}[1]{\ensuremath{#1\text{\normalfont -EDCS}}}

Although a maximum matching can be found in polynomial time \cite{edmonds1965paths}, there are many natural settings and computational models in which an exact maximum matching is provably impossible to find. As a result, the {\em approximate} maximum matching problem has also received significant attention. For $0 \leq \alpha \leq 1$ a matching is said to provide an $\alpha$-approximation if its size is at least $\alpha$ fraction of the size of a maximum matching. While different ad hoc algorithms for approximate matching have been developed in various computational models, generic algorithmic tools that are applicable to a wider range of settings have been of special interest. The {\em edge-degree constrained subgraph} (EDCS), introduced by Bernstein and Stein \cite{BernsteinSteinICALP15}, is one of the most successful such methods. Since its introduction, the EDCS has found applications in dynamic graph algorithms, communication complexity, stochastic matching, fault tolerant matching, and random-order streaming algorithms 
\cite{BernsteinSteinICALP15,BernsteinSteinSODA16,AssadiBBMSSODA19,AssadiBernsteinSOSA19,BernsteinICALP20,AssadiB21} (see specifically the paper of \cite{AssadiBernsteinSOSA19} which discusses many of these applications). Motivated by its success and versatility, we revisit EDCS in this work and further study its {\em key property} in general graphs.

The EDCS is formally defined in the following simple and clean way:

\begin{definition}[\cite{BernsteinSteinICALP15}]
	Let $G=(V, E)$ be a graph and let $\beta > \beta^- \geq 1$ be integers. A subgraph $H=(V, E_H)$ of $G$ is a {\em $(\beta, \beta^-)$-EDCS} of $G$ if it satisfies the following two properties:
	\begin{enumerate}[label={\normalfont (P\arabic*)}, itemsep=0pt, topsep=5pt]
		\item For any edge $e = (u, v) \in E_H$, \,\,\,\,\,\,\,\,\,\,\,\,\,\,$\deg_H(u) + \deg_H(v) \leq \beta$.
		\item For any edge $e = (u, v) \in E \setminus E_H$, \,\,\,\,$\deg_H(u) + \deg_H(v) \geq \beta^-$.
	\end{enumerate}
\end{definition}

It is easy to verify that for $\beta = 2$ and $\beta^- = 1$, which are the smallest integers satisfying $\beta > \beta^- \geq 1$, subgraph $H$ is a \EDCS{(2, 1)} if and only if it is a {\em maximal matching} of $G$ (where recall that a matching is maximal if it is not a subset of another matching). As such, the EDCS can be seen as a natural generalization of maximal matchings. Somewhat surprisingly, however, by slightly increasing $\beta$ and $\beta^-$, the approximation achieved by EDCS gets much better than what maximal matchings achieve. Specifically, while a maximal matching provides only a $\frac{1}{2}$-approximation\footnote{Consider a path of length three with the middle edge forming a maximal matching.}, the key property of EDCS  asserts that for large enough $\beta$ and $\beta^-$ close enough to $\beta$, the EDCS includes an (almost) $\frac{2}{3}$-approximate maximum matching of $G$.

This key property of EDCS was first proved by Bernstein and Stein \cite{BernsteinSteinICALP15} for {\em bipartite} graphs, obtaining that $\beta \geq 1/\epsilon$ and $\beta^- \geq (1-\epsilon) \beta$ suffice for a $(\frac{2}{3} - O(\epsilon))$-approximation. Note that since all the edges in a \EDCS{(\beta, \beta^-)} have {\em edge-degree} bounded by $\beta$, all the vertices also have degree at most $\beta$. Therefore, a \EDCS{(\beta, \beta^-)} has at most $O(n \beta)$ edges. For this reason, the smaller values of $\beta$ are more desirable as they make the EDCS sparser.

The same approximation is also known to be achievable by EDCS in {\em general} graphs, albeit with a stronger assumption on $\beta$. Particularly, \cite{BernsteinSteinSODA16} proved via the probabilistic method that $\beta \geq 1/\epsilon^{3}$ and $\beta^- \geq (1-\epsilon) \beta$ suffice for a $(\frac{2}{3} - O(\epsilon))$-approximation in general graphs. By another proof based on the probabilistic method and the Lov\'asz Local Lemma,  Assadi and Bernstein \cite{AssadiBernsteinSOSA19} proved that $\beta \geq \log(1/\epsilon)/ \epsilon^{2}$ and $\beta^- \geq (1-\epsilon) \beta$ also suffice for a $(\frac{2}{3} - O(\epsilon))$-approximation, improving the cubic dependence of \cite{BernsteinSteinSODA16} on $1/\epsilon$ to near quadratic. While the analysis of \cite{AssadiBernsteinSOSA19} is particularly clean and simple, the near-quadratic dependence on $1/\epsilon$ is inherent to it, leaving a quadratic gap between our understanding of EDCS in general and bipartite graphs.

In this work, we close this gap between general and bipartite graphs, by improving the dependence of $\beta$ on $1/\epsilon$ in general graphs to linear. We achieve this via a new proof of the approximation guarantee that is based on the graph's Gallai-Edmonds decomposition and completely avoids the probabilistic method of the previous proofs. Formally, using $\mu(H')$ to denote the size of a maximum matching of $H'$, we prove the following guarantee:

\begin{graytbox}
\begin{theorem}\label{thm:main}
	Fix any $0 < \epsilon \leq 1$ and let $\beta \geq 50/\epsilon$ and $\beta^- \geq (1-\frac{\epsilon}{10})\beta$ be integers. For any graph $G$ and any \EDCS{(\beta, \beta^-)} $H$ of $G$ it holds that $\mu(H) \geq \big(\frac{2}{3} - \epsilon \big) \cdot \mu(G).$
\end{theorem}
\end{graytbox}

\Cref{thm:main} leads to improvements in many applications of EDCS. Here we briefly mention one of these applications which is in the {\em one-way communication complexity} model: Alice and Bob are each given a subset of the edges of a graph $G$. The goal is for Alice to send a small message to Bob such that Bob can report a large matching of $G$. This communication problem has been studied extensively due to its applications in the streaming setting. It is known that a $2/3$-approximation is the best approximation achievable via  $\widetilde{O}(n)$ communication protocols \cite{DBLP:conf/soda/GoelKK12}. Additionally,  \cite{AssadiBernsteinSOSA19} first showed that EDCS can be employed to obtain a $(\frac{2}{3} - \epsilon)$-approximation for general graphs with $O(n \log(1/\epsilon)/\epsilon^2)$ communication. Plugging \Cref{thm:main} in the analysis of \cite[Section~4]{AssadiBernsteinSOSA19} immediately improves the communication to $O(n/\epsilon)$:

\begin{corollary}
	There is a deterministic poly-time one-way communication complexity protocol that given $\epsilon > 0$, finds a $(\frac{2}{3} - \epsilon)$-approximate matching with $O(n/\epsilon)$ communication from Alice to Bob.
\end{corollary}

%% file: thealgorithm.tex
\section{Preliminaries}

Given a graph $G = (V, E)$ and a vertex subset $U \subseteq V$, we use $G[U]$ to denote the subgraph of $G$ induced on $U$. For any vertex $v$ we use $\deg_G(u)$ to denote the degree of $v$ in graph $G$, and for any edge $e=(u, v)$ we use $\deg_G(e)$ to denote $\deg_G(u) + \deg_G(v)$; when graph $G$ is clear from context, we may drop the subscript. For any integer $k$, we use $[k]$ to denote set $\{1, \ldots, k\}$.

We will need the following fact.

\begin{fact}\label{fact:1/1+x} 
	For any $x \geq 0$, $\frac{1 - x}{1 + x} \geq 1-2x$.
\end{fact}

We also use the following property of edge-degree bounded bipartite graphs implicit in \cite{AssadiBernsteinSOSA19,BernsteinSteinICALP15}. The following proof is particularly similar to the proof of \cite[Proposition~2.6]{AssadiBernsteinSOSA19}.

\begin{lemma}\label{lem:partsizeEDCS}
	Let $G=(P, Q, E)$ be a bipartite graph with every edge $e \in E$ satisfying $\deg_G(e) \leq \beta$. Let $d_P := |E|/|P|$ be the average degree of the vertices in $P$; then
	$
	|Q| \geq \frac{d_P}{\beta - d_P} |P|.
	$
\end{lemma}
\begin{myproof}
	Let $d_Q := |E|/|Q|$ and observe that $|Q| d_Q = |P| d_P = |E|$. This implies that $|Q| = \frac{d_P}{d_Q}|P|$ and so it suffices to prove $d_Q \leq \beta - d_P$. Toward proving this, first note that
	$$
		\sum_{(u, v) \in E} \deg(v) + \deg(u) = \sum_{v \in P} \deg(v)^2 + \sum_{u \in Q} \deg(u)^2 \geq \sum_{v \in P} \left(\frac{|E|}{|P|}\right)^2 + \sum_{u \in Q} \left(\frac{|E|}{|Q|}\right)^2 = |E| (d_P + d_Q),
	$$
	where the inequality follows from the fact that $\sum_{v \in P} \deg(v) = \sum_{u \in Q} \deg(u) = |E|$, and both $\sum_{v \in P} \deg(v)^2$ and $\sum_{u \in Q} \deg(u)^2$ are minimized when the summands are equal. 
	
	Next, observe that $\sum_{(u, v) \in E} \deg(v) + \deg(u) \leq \beta |E|$ by property (P1) of EDCS. Combined with inequality above, this gives $|E|(d_P + d_Q) \leq \beta|E|$ which proves $d_Q \leq \beta - d_P$.
\end{myproof}

\paragraph{Gallai–Edmonds Decomposition.}
Consider any graph $G=(V, E)$. The Gallai-Edmonds decomposition for $G$ is a partitioning of the vertex set $V$ into three subsets $D(G), A(G), C(G)$ defined as follows. The set $D(G)$ includes a vertex $v$ iff there exists some maximum matching of $G$ that does not match $v$. The set $A(G)$ includes any vertex not in $D(G)$ that is adjacent to a vertex in $D(G)$. The set $C(G)$ includes all the rest of vertices, i.e., those that are not in $D(G)$ or $A(G)$. 

The following proposition summarizes some of the properties of the decomposition that we use in this paper. See \cite[Theorem~3.2.1]{MatchingTheoryBook} for more properties and the proof.

\begin{proposition}\label{prop:GE}
	Let $G$ be any graph and let $D(G), A(G), C(G)$ be defined as above. Fix any arbitrary maximum matching $M$ of $G$; it must satisfy all the following:
	\begin{enumerate}[label=$(\arabic*)$, itemsep=0pt, topsep=0pt]
		\item Every vertex in $C(G)$ is matched to another vertex in $C(G)$.
		\item Every vertex in $A(G)$ is matched to a vertex in $D(G)$.
		\item Every connected component $O$ of $G[D(G)]$ has exactly one vertex that is either unmatched or matched to $A(G)$; all the other $|O|-1$ vertices of $O$ are matched together. Note, in particular, that this implies all connected components of $G[D(G)]$ are of odd size.
	\end{enumerate}
\end{proposition}

\section{The New Proof of the Key Property of EDCS}

Let $H$ be a $(\beta, \beta^-)$-EDCS of a graph $G$ where $\beta$ and $\beta^-$ satisfy the conditions of \Cref{thm:main}. We prove \Cref{thm:main} in this section that $\mu(H) \geq (\frac{2}{3} - \epsilon)\mu(G)$. 

In our proof, we apply the Gallai-Edmonds decomposition on graph $H$. Let us for brevity use $D := D(H), A := A(H),$ and $C := C(H)$ to refer to the decomposition. Fix an arbitrary maximum matching $M$ of $H$, and an arbitrary maximum matching $M^\star$ of $G$. We call a vertex {\em special} if it belongs to $D$ and is either unmatched by $M$ or is matched by $M$ to a vertex in $A$. We let $S \subseteq V$ denote the set of special vertices. \Cref{prop:GE} part (3) immediately implies the following.

\begin{observation}\label{obs:special}
	Every connected component in $H[D]$ has exactly one special vertex.	
\end{observation}

Since $M$ and $M^\star$ are matchings, their union is a collection of paths and cycles whose edges alternate between the two matchings. Define an {\em augmenting path} to be any path in $M \cup M^\star$ that starts and ends with an edge of $M^\star$, and let $\mc{P}$ be the set of all augmenting paths.

\begin{observation}\label{obs:num-of-augpaths}
	$|\mc{P}| \geq \mu(G) - \mu(H)$.
\end{observation}
\begin{myproof}
	For any augmenting path $P \in \mc{P}$, $|P \cap M^\star| = |P \cap M| + 1$ and for any other component $C \in (M \cup M^\star) \setminus \mc{P}$, $|C \cap M^\star| \leq |C \cap M|$. Hence, $|\mc{P}| \geq |M^\star| - |M| = \mu(G) - \mu(H)$.
\end{myproof}

\Cref{lem:Pstructure} (see also \Cref{fig:Pstructure}) is our most technical lemma and is the key to our analysis.

\begin{lemma}\label{lem:Pstructure}
	Take any augmenting path $P \in \mc{P}$. Call an edge $e=(u, v)$ of $P$ {\em suitable} if $e \not\in H$, $u \not\in A$, and $v \not\in A$. At least one of the following should hold for $P$:
	\begin{enumerate}[label={\normalfont (T\arabic*)}, itemsep=0pt, topsep=0pt]
		\item $P$ includes a suitable edge $e$ such that the endpoints of $e$ belong to two distinct connected components $O_1$ and $O_2$ of $H[D]$ and the special vertices of $O_1$ and $O_2$ both belong to $P$.
		\item $P$ includes two vertex disjoint suitable edges.
	\end{enumerate}
\end{lemma}

\begin{figure}[h]
  \centering
  \includegraphics[width=\textwidth]{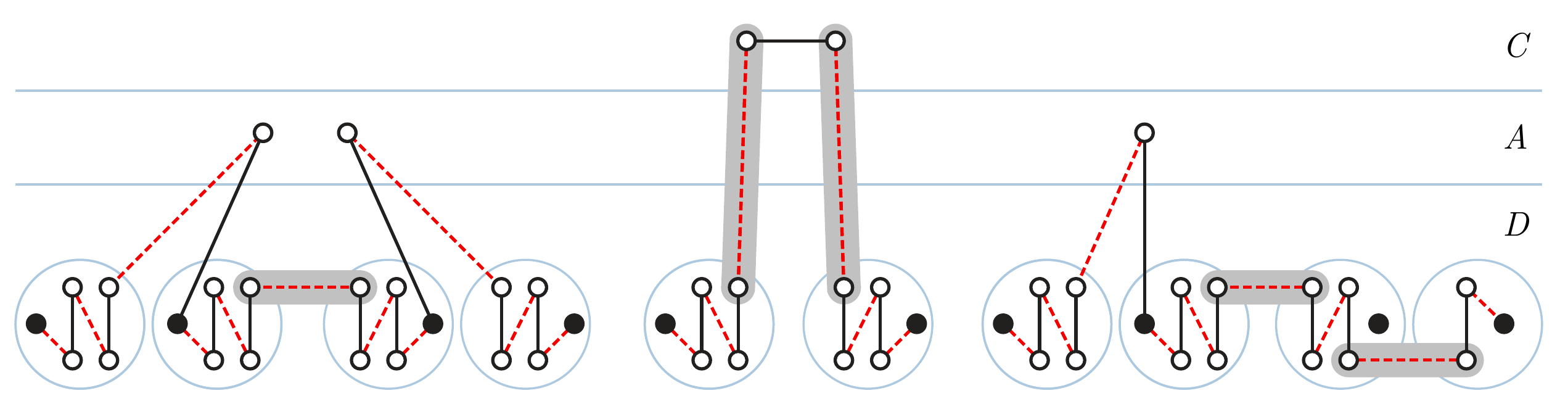}
  \caption{An illustration of how augmenting paths satisfy \Cref{lem:Pstructure}. Here the Gallai-Edmonds partitioning $C, A, D$ of $H$, as well as the connected components of $H[D]$ are shown. The black edges are the edges of $M$, the dashed (red) edges are the edges of $M^\star$, and special vertices are colored black. Edges highlighted with gray are the ``suitable'' edges defined in \Cref{lem:Pstructure}. The augmenting path on the left has one suitable and satisfies \Cref{lem:Pstructure}-(T1). The other two augmenting paths have two suitable edges each, and satisfy \Cref{lem:Pstructure}-(T2).}
  \label{fig:Pstructure}
\end{figure}

\begin{proof}
	Let $W=(w_1, \ldots, w_k)$ be obtained from $P$ after contracting all the vertices of each connected component in $H[D]$ and removing the resulting self-loops. Note that $W$ is a walk but not necessarily a path, since $P$ may enter and leave a component of $H[D]$ multiple times. We use $e_i$ to denote the edge in $P$ that corresponds to the edge $(w_i, w_{i+1})$ of $W$. Also, with a slight abuse of notation, we say $w_i \in D$ if $w_i$ is a contracted component of $H[D]$.
	
	From the definition of augmenting paths, the endpoint vertices of $P$ must be unmatched by $M$. Recall from \Cref{prop:GE} that every vertex unmatched by $M$ must belong to a unique connected component of $H[D]$. As such, $w_1 \not= w_k$ and $w_1, w_k \in D$.
	
	We call index $i \in [k-1]$ {\em good} if $(w_i \in D, w_{i+1} \in D)$ or $(w_i \in D, w_{i+1} \in C)$ or $(w_i \in C, w_{i+1} \in D)$.
	
	\begin{claim}\label{cl:clllc39-32}	
		For any good index $i$,  $e_i$ is suitable.
	\end{claim}
	\begin{myproof}
	By definition of a good index $i$, the endpoints of $e_i$ already do not belong to $A$. So it remains to show that $e_i \not\in H$. If $w_i \in D, w_{i+1} \in D$, then $e_i$ connects two distinct components of $H[D]$ and so $e_i$ cannot belong to $H$. If $(w_i \in D, w_{i+1} \in C)$ or $(w_i \in C, w_{i+1} \in D)$, then $e_i$ connects a vertex in $D$ to a vertex in $C$. By the construction of Gallai-Edmonds decomposition, $H$ has no edges between $D$ and $C$, implying that $e_i \not\in H$.
	\end{myproof}
	
	If two distinct indices $i$ and $i'$ are good, then by \Cref{cl:clllc39-32} $e_i$ and $e_{i'}$ are suitable. Furthermore, $e_i$ and $e_{i'}$ are vertex disjoint since otherwise one of them, say $e_i$, must belong to $M$ since $P$ alternates between $M$ and $M^\star$. Since $M \subseteq H$, this implies $e_i \in H$ which contradicts $e_i$ being suitable. Hence, if we have two distinct good indices, \Cref{lem:Pstructure}-(T2) is satisfied and we are done. Let us therefore assume for the rest of the proof that at most one index $i$ is good.
	
	If $W$ has size $k=2$, then since $w_1, w_2 \in D$, \Cref{cl:clllc39-32} implies $e_1$ is suitable. Additionally, $e_1$ connects two connected components of $H[D]$ whose special vertices are unmatched in $M$ and belong to $P$ (the two endpoints of $P$). Therefore, $e_1$ satisfies \Cref{lem:Pstructure}-(T1). So let us also assume for the rest of the proof that $k > 2$.
	
	With the two assumptions above that $k > 2$ and there is at most one good index $i$, we get that $w_2 \in A$ or $w_{k-1} \in A$ as otherwise both $1$ and $k-1$ would be good indices. We assume w.l.o.g. that $w_2 \in A$ noting that the other case is equivalent if one relabels $W$ in the reverse direction. 
	
	Recall from \Cref{obs:special} that every component in $H[D]$ has exactly one special vertex where a special vertex is one that is either unmatched in $M$ or matched to a vertex in $A$. Since the special vertex in the component $w_1$ is unmatched by $M$ (as it starts the augmenting path $P$) there is no edge in $M$ from $w_1$ to $w_2 \in A$. This means that $e_1 \in M^\star$. Since $P$ alternates between $M^\star$ and $M$ edges, the next edge $e_2$ must belong to $M$. By \Cref{prop:GE}-(2),(3) every vertex in $A$ is matched to a distinct connected component of $H[D]$ in $M$. As such, $w_3 \in D$ and the special vertex of component $w_3$ belongs to $P$ and is matched in $M$ to $A$. Since all the vertices of the component $w_3$ except the special vertex are matched together by \Cref{prop:GE}-(3), we get that edge $e_3$ which leaves the component $w_3$ cannot belong to $M$ and so belongs to $M^\star$. Now if index $3$ is not good, we must have $w_4 \in A$ by definition of good indices. Continuing this argument and denoting by $\ell$ the first index that is either good or $\ell = k$, we get that $\ell$ must be odd and
	\begin{flalign}
	\label{eq:lrcc-390019}&w_1 \in D, w_2 \in A, w_3 \in D, w_4 \in A,\ldots, w_{\ell-1} \in A, w_\ell \in D, \\
	\nonumber &e_1 \in M^\star, e_2 \in M, e_3 \in M^\star, \ldots, e_{\ell-1} \in M, \text{ and}\\
	\label{eq:hcrllrc-921308}&\text{the special vertices of $w_1, w_3, \ldots, w_\ell$ belong to $P$.}
	\end{flalign}
	It cannot hold that $\ell =k$ since the special vertex of $w_\ell$ is matched in $M$ to $w_{\ell-1} \in A$ by (\ref{eq:lrcc-390019}) whereas the special vertex of $w_k$ must be unmatched by $M$ since $P$ is an augmenting path. Hence, $\ell < k$ and $\ell$ must be our unique good index. We also claim that 
	\begin{flalign}
	\label{eq:rgfb-08873129}&w_k \in D, w_{k-1} \in A, w_{k-2} \in D, w_{k-3} \in A,\ldots, w_{\ell+1} \in D, \\
	\nonumber &e_{k-1} \in M^\star, e_{k-2} \in M, e_{k-3} \in M^\star, \ldots, e_{\ell} \in M^\star, \text{ and}\\
	\label{eq:hbggr-9218379}&\text{the special vertices of $w_k, w_{k-2}, w_{k-4}, \ldots, w_{\ell+1}$ belong to $P$.}
	\end{flalign}
	This trivially holds if $\ell = k-1$. If $\ell < k-1$, then $k-1$ is not a good index and so we can apply the same argument above now starting from $w_k$ and traversing $W$ in the opposite direction.
	
	We prove $e_\ell$ satisfies \Cref{lem:Pstructure}-(T1). First, since $\ell$ is good, from \Cref{cl:clllc39-32} we get that $e_\ell$ is suitable. We also have $w_\ell \in D$ by (\ref{eq:lrcc-390019}) and $w_{\ell+1} \in D$ by (\ref{eq:rgfb-08873129}), and the special vertices of $w_{\ell}$ and $w_{\ell+1}$ belong to $P$ by (\ref{eq:hcrllrc-921308}) and (\ref{eq:hbggr-9218379}). Hence, $e_\ell$ satisfies \Cref{lem:Pstructure}-(T1) concluding the proof.
\end{proof}

We use $\mc{P}_1 \subseteq \mc{P}$ to denote the augmenting paths satisfying \Cref{lem:Pstructure}-(T1) and define $\mc{P}_2 = \mc{P} \setminus \mc{P}_1$. Note that for any $P \in \mc{P}_2$ \Cref{lem:Pstructure}-(T2) must hold.

Next, we construct an auxiliary bipartite graph $B=(W, Z, E_B)$ that we use to lower bound $\mu(H)$. The vertex parts $W$ and $Z$ of $B$ are both subsets of $V$ but, importantly, a vertex in $V$ may belong to both $W$ and $Z$ --- hence, $B$ is not simply a bipartite subgraph of $H$ or $G$.

\begin{highlighttechnical}
\textbf{Construction of $B = (W, Z, E_B)$:}
\begin{itemize}[leftmargin=20pt,itemsep=0pt]
	\item Part $W$ is composed of two disjoint subsets $W_1$ and $W_2$ with $|W_1| = 2|\mc{P}_1|$ and $|W_2| = 4|\mc{P}_2|$. Particularly, $W_1$ for every $P \in \mc{P}_1$, includes the two endpoints of a suitable edge $e$ of $P$ that satisfies \Cref{lem:Pstructure}-(T1). Moreover, $W_2$ for every $P \in \mc{P}_2$, includes the 4 endpoints of the two vertex-disjoint suitable edges in $P$ satisfying \Cref{lem:Pstructure}-(T2).
	\item Part $Z$ includes the vertices in $V \setminus S$. We use $Z_A$ to denote $Z \cap A$ which equals $A$.
	\item For any $w_2 \in W_2$ and any $z \in Z$ such that $(w_2, z) \in H$, we add the edge $(w_2, z)$ to $E_B$.
	\item For any $w_1 \in W_1$ and any $z \in Z_A$ such that $(w_1, z) \in H$, we add the edge $(w_1, z)$ to $E_B$.
\end{itemize}
\end{highlighttechnical}

A few definitions are in order. Define $\lambda := \epsilon/10$; the statement of \Cref{thm:main} implies
\begin{equation}\label{eq:lambda}
	0 < \lambda \leq 1/10, \qquad \beta \geq 5/\lambda, \qquad \beta^- \geq (1-\lambda)\beta \geq 4/\lambda.
\end{equation}
We define $\alpha := |W_1|/|W|$ and $\sigma := \frac{2|D \setminus S|}{|W_1|\beta^-}$ for simplifying our bounds. Finally, for any $U \subseteq W$, we use $\bar{d}_B(U)$ to denote the average degree of $U$ in graph $B$. In other words, $\bar{d}_B(U) \cdot |U|$ is the number of edges of $B$ with one endpoint in $U$.

The next useful claim relates the maximum matching size of $H$ to part $Z$ of graph $B$.

\begin{lemma}\label{cl:muH-from-Z}
	$\mu(H) \geq (\frac{1}{2} - \lambda) |Z| + \frac{1}{2}|Z_A| + \lambda |D \setminus S|$.
\end{lemma}
\begin{myproof}
	Let us define $\phi(v) : V \to \{0, \sfrac{1}{2}, 1\}$ as follows:  For any vertex $v \in A$ let $\phi(v) = 1$, for any vertex $v \in S$ let $\phi(v) = 0$, and for all other vertices $v$ let $\phi(v) = \sfrac{1}{2}$. 
	
	First, we show that $\sum_{v \in V} \phi(v) = |M| = \mu(H)$. Since any vertex $v$ unmatched by $M$ must be special (by \Cref{prop:GE} and definition of special vertices) and so $\phi(v) = 0$, it suffices to prove that for any edge $(u, v) \in M$, $\phi(u) + \phi(v) = 1$. Indeed, if one of $u$ and $v$ belongs to $A$ the other must be special (by \Cref{prop:GE} and definition of special vertices) and so $\phi(u)  + \phi(v) = 1+ 0 = 1$. Otherwise, neither of $u$ or $v$ is in $A$ or is special, and thus $\phi(u) + \phi(v) = \sfrac{1}{2} + \sfrac{1}{2} = 1$.
	
	Next, observe from definition of $\phi$ that $\sum_{v \in V} \phi(v) = \frac{1}{2}|V \setminus (A \cup S)| + |A| = \frac{1}{2}|V \setminus S| + \frac{1}{2} |A|$. Since by construction $Z = V \setminus S$ and  $Z_A = Z \cap A = A$, we get $\sum_{v \in V} \phi(v) = \frac{1}{2} |Z| + \frac{1}{2}|Z_A|$. Thus,
	$$\mu(H) = \sum_{v \in V} \phi(v) = \frac{1}{2}|Z| + \frac{1}{2}|Z_A| = \Big(\frac{1}{2} - \lambda\Big)|Z| + \lambda |Z| + \frac{1}{2}|Z_A| \geq \Big(\frac{1}{2} - \lambda \Big)|Z| + \frac{1}{2}|Z_A| + \lambda|D \setminus S|,$$ 
	where last inequality follows from $Z = V \setminus S \supseteq D \setminus S$ which implies $|Z| \geq |D \setminus S|$.
\end{myproof}

From \Cref{cl:muH-from-Z}, we get that to lower bound $\mu(H)$ it suffices to lower bound $|Z|$ and $|Z_A|$. Our goal is to use \Cref{lem:partsizeEDCS} for this purpose; so let us upper bound the edge-degrees of $B$ first, and then lower bound the average degrees of $W_1$ and $W_2$ in \Cref{cl:avgdegW2} and \Cref{cl:avgdegW1}.

\begin{observation}\label{obs:edge-degreeB}
	For any edge $e \in E_B$, $\deg_B(e) \leq \beta$.
\end{observation}
\begin{myproof}
	By construction of  $B$, $\deg_H(w) \geq \deg_B(w)$ for any vertex $w \in W$ and  $\deg_H(z) \geq \deg_B(z)$ for any vertex $z \in Z$. Additionally, if there is an edge $e = (w, z) \in B$, then $(w, z) \in H$ too and so $\deg_H(e) \leq \beta$. As such, $\deg_B(e) = \deg_B(w) + \deg_B(z) \leq \deg_H(w) + \deg_H(z) = \deg_H(e) \leq \beta$.
\end{myproof}

\begin{lemma}\label{cl:avgdegW2}
	$\bar{d}_B(W_2) \geq (1- \lambda)\beta/2$.
\end{lemma}
\begin{myproof}
	Take any augmenting path $P \in \mc{P}_2$ and let $e = (u, v)$, $e'=(u', v')$ be the edges in $P$ that satisfy \Cref{lem:Pstructure}-(T2) and are used in defining $W_2$. Since $e, e' \not\in H$ by \Cref{lem:Pstructure}, it must hold that $\deg_H(e) \geq \beta^-$ and $\deg_H(e') \geq \beta^-$ by property (P2) of EDCS. As such, at least $2 \beta^-$ edges must be connected to the four distinct endpoints of $e, e'$ in $H$. Additionally, for any $w_2 \in W_2$ we have $\deg_H(w_2) = \deg_B(w_2)$ by construction of $B$. As such, $\deg_B(u) + \deg_B(v) + \deg_B(u') + \deg_B(v') \geq 2\beta^-$. Hence, overall $2 \beta^- |\mc{P}_2| = \frac{\beta^-}{2} |W_2|$ edges are connected to $W_2$ in $B$, implying that $\bar{d}_B(W_2) \geq \beta^-/2 \geq (1-\lambda)\beta/2$ where the last inequality follows from (\ref{eq:lambda}).
\end{myproof}

\begin{lemma}\label{cl:avgdegW1}
	$\bar{d}_B(W_1) \geq (1-\sigma - \lambda)\beta/2$.
\end{lemma}

To prove \Cref{cl:avgdegW1}, let us first prove two auxiliary claims.

\begin{claim}\label{cl:uniqueW1}
		Every connected component $O$ in $H[D]$ has at most one vertex in $W_1$.
\end{claim}
\begin{myproof}
	Suppose toward contradiction that two distinct vertices $w_1, w'_1 \in W_1$ belong to the same connected component $O$ in $H[D]$. By \Cref{lem:Pstructure}-(T1) the augmenting path that includes $w_1$ must include the special vertex of $O$, and the same also holds for $w'_1$. Hence, $w_1$ and $w'_1$ must belong to the same augmenting path $P \in \mc{P}_1$. Hence, $w_1$ and $w'_1$ must be the two endpoints of the suitable edge satisfying \Cref{lem:Pstructure}-(T1) for $P$. But \Cref{lem:Pstructure}-(T1) also guarantees that the two endpoints of this suitable edge belong to different connected components of $H[D]$, a contradiction.
\end{myproof}

\begin{claim}\label{cl:hclrc9-12300}
	$\sum_{w_1 \in W_1} \deg_B(w_1) \geq \sum_{w_1 \in W_1} \deg_H(w_1) - |D \setminus S|.$
\end{claim}
\begin{proof}	
	Let $F = \{(w_1, v) \in H \mid w_1 \in W_1, v \not\in A\}$. We claim that $(i)$ $v \not\in W_1$ for any $(w_1, v) \in F$, and $(ii)$ $|F| \leq |D \setminus S|$. Toward this, take an edge $(w_1, v) \in F$. By definition of $W_1$, $w_1$ is an endpoint of an edge satisfying \Cref{lem:Pstructure}-(T1) which asserts $w_1$ belongs to some connected component $O$ of $H[D]$. Given that $(w_1, v) \in H$ and $v \not\in A$, by Gallai-Edmonds it must be that $v \in O$ also. So immediately from \Cref{cl:uniqueW1} we get that $v \not\in W_1$ proving $(i)$. To prove $(ii)$, for any edge $(w_1, v) \in F$ we charge the component $O$ of $H[D]$ where $w_1 \in O$. Observe that $w_1$ has at most $|O|-1$ edges in $O$ and so charges $O$ at most $|O|-1$ times. Indeed, $O$ is in total charged at most $|O|-1$ times since $w_1$ is its only vertex in $W_1$ by \Cref{cl:uniqueW1}. Hence, the total number of charges, which equals $|F|$, is upper bounded by $\sum_O (|O|-1)$ where $O$ ranges over all connected components of $H[D]$. Since each component $O$ of $H[D]$ has exactly $|O|-1$ non-special vertices by \Cref{obs:special}, this indeed equals $|D \setminus S|$ and so $|F| \leq |D \setminus S|$.

	Now take $w_1 \in W_1$ and consider an edge $(w_1, v) \in H$. Note that if $v \in A$ then we also have an edge $(w_1, v)$ in $B$ by construction. If $v \not\in A$, on the other hand, we do not have a corresponding edge to $(w_1, v)$ in $B$. However, in this case $(w_1, v) \in F$ and as discussed above, $v \not\in W_1$ and $|F| \leq |D \setminus S|$. Hence, $\sum_{w_1 \in W_1} \deg_B(w_1) \geq \sum_{w_1 \in W_1} \deg_H(w_1) - |D \setminus S|.$
\end{proof}
	
\begin{proof}[Proof of \Cref{cl:avgdegW1}]
	Take an arbitrary augmenting path $P \in \mc{P}_1$ and let $(u, v) \in P$ be the suitable edge of $P$ used in defining $W_1$.  Since $e = (u, v) \not\in H$ by definition of suitable edges, $\deg_H(e) \geq \beta^-$ by property (P2) of EDCS. Hence, $\sum_{w_1 \in W_1} \deg_H(w_1) \geq \frac{1}{2}|W_1|\beta^-$. Thus by \Cref{cl:hclrc9-12300},
	$$
	\sum_{w_1 \in W_1} \deg_B(w_1) \geq \frac{1}{2}|W_1|\beta^- - |D \setminus S| \stackeq{$\left(\sigma = \frac{2|D \setminus S|}{|W_1|\beta^-}\right)$}{=} \frac{1}{2}|W_1|\beta^- - \frac{1}{2} \sigma |W_1|\beta^- = (1 - \sigma)|W_1|\beta^-/2.
	$$
	Therefore,
	$
	\bar{d}_B(W_1) \geq (1-\sigma)\beta^-/2 \geq (1-\sigma)(1-\lambda)\beta/2 \geq (1-\sigma - \lambda)\beta/2.
	$
\end{proof}

Next, we lower bound $|Z|$ by applying \Cref{lem:partsizeEDCS} on $B$ and using the lower bounds of \Cref{cl:avgdegW1} and \Cref{cl:avgdegW2} on the average degrees of $W_1, W_2$.

\begin{lemma}\label{cl:lbZ}
	$|Z| \geq (1-2\alpha \sigma - 2\lambda)|W|$.
\end{lemma}
\begin{myproof}
	We prove this by applying \Cref{lem:partsizeEDCS} on graph $B$. For this, we need to lower bound the average degree of the $W$ part of the graph. We have
	\begin{flalign}
	\nonumber \bar{d}_B(W) &= \frac{\bar{d}_B(W_1)|W_1| + \bar{d}_B(W_2)|W_2|}{|W|} \geq \frac{(1 - \sigma - \lambda) \frac{\beta}{2} \cdot \alpha |W| + (1-\lambda) \frac{\beta}{2} \cdot (1-\alpha)|W|}{|W|}\\
	&= \Big((1 - \sigma - \lambda)\alpha + (1-\lambda)(1-\alpha)\Big)\frac{\beta}{2} = (1 - \sigma \alpha - \lambda) \frac{\beta}{2}.\label{eq:llch-9009812388}
	\end{flalign}
	Now since $B$ has edge-degree $\leq \beta$ by \Cref{obs:edge-degreeB}, we can apply \Cref{lem:partsizeEDCS} on graph $B$ to get
$$
	|Z| \geq \frac{\bar{d}_B(W)}{\beta - \bar{d}_B(W)} |W|
	\stackeq{(\ref{eq:llch-9009812388})}{\geq}
	\frac{(1-\sigma \alpha - \lambda) \beta/2}{\beta - (1 - \sigma \alpha - \lambda)\beta/2} |W|
	= \frac{1 - \sigma \alpha  - \lambda}{1 +  \sigma \alpha + \lambda} |W|
	\stackeq{\Cref{fact:1/1+x}}{\geq}
	(1-2\alpha \sigma - 2\lambda)|W|.\qedhere
$$
\end{myproof}

Now we lower bound $|Z_A|$ again by applying \Cref{lem:partsizeEDCS}, but this time on subgraph $B[W_1 \cup Z_A]$.

\begin{lemma}\label{cl:lbZA}
	$|Z_A| \geq (1-2\sigma - 2\lambda) |W_1|$.
\end{lemma}
\begin{myproof}
	We apply \Cref{lem:partsizeEDCS} on the induced subgraph $H[W_1 \cup Z_A]$ of $H$. Observe that since by construction of $B$ all the edges of $W_1$ go to subset $Z_A$ of $Z$, then $\bar{d}_B(W_1)$ is also the average degree of $W_1$ in $H[W_1 \cup Z_A]$. Hence, we get from \Cref{lem:partsizeEDCS} that
	$$
		|Z_A| \geq \frac{\bar{d}_B(W_1)}{\beta - \bar{d}_B(W_1)}|W_1|
		\stackeq{\Cref{cl:avgdegW1}}{\geq}
		\frac{(1-\sigma - \lambda)\beta/2}{\beta - (1-\sigma - \lambda)\beta/2} |W_1|
		= \frac{1 - \sigma - \lambda}{1 + \sigma + \lambda} |W_1|
		\stackeq{\ref{fact:1/1+x}}{\geq} (1-2\sigma-2\lambda)|W_1|.\qedhere
	$$
\end{myproof}

We are now ready to finish the proof of \Cref{thm:main}.

\begin{proof}[Proof of \Cref{thm:main}]
We have
\begin{flalign*}
	\mu(H) &\geq \Big(\frac{1}{2} - \lambda\Big)(|Z| + |Z_A|) + \lambda |D \setminus S| \tag{By \Cref{cl:muH-from-Z}.}\\
	&\geq \Big(\frac{1}{2} - \lambda\Big)\Big((1-2\alpha \sigma - 2\lambda)|W| + (1-2\sigma - 2\lambda)\alpha|W|\Big) + \lambda|D \setminus S| \tag{By \Cref{cl:lbZ} and \Cref{cl:lbZA}.}\\
	&= \Big(\frac{1}{2} - \lambda\Big)\Big((1- \alpha)|W| + (2-4\sigma)\alpha|W| - 2\lambda(1+\alpha)|W|\Big) + \lambda|D \setminus S|\\
	&\geq \Big(\frac{1}{2} - \lambda \Big)\Big(|W_2| + (2-4\sigma)|W_1| - 4\lambda|W|\Big) + \lambda|D \setminus S| \tag{Since $|W_1| = \alpha|W|, |W_2| = (1-\alpha)|W|,$ and $\alpha \leq 1$.}\\
	&\geq \frac{1}{2}|W_2| + (1 - 2 \sigma) |W_1|  - 5 \lambda |W|  + \lambda|D \setminus S|\\
	&= \frac{1}{2}|W_2| + (1 - 2 \sigma) |W_1|  - 5 \lambda |W|  + \frac{1}{2}\lambda \sigma |W_1| \beta^- \tag{Since $\sigma = \frac{2|D \setminus S|}{|W_1| \beta^-}$.}\\
	&= \frac{1}{2}|W_2| + |W_1|  - 5 \lambda |W| \tag{Since $\beta^- \geq 4/\lambda$ by (\ref{eq:lambda}).}
\end{flalign*}
Now recalling from the definition of $W_1$ and $W_2$ that $|W_2| = 4|\mc{P}_2|$ and $|W_1| = 2|\mc{P}_1|$, and noting that $|W| = |W_1| + |W_2| \leq 4 |\mc{P}|$, we get
$$
	\mu(H) \geq \frac{1}{2} \cdot 4|\mc{P}_2| + 2|\mc{P}_1| - 20 \lambda |\mc{P}| = (2 - 20 \lambda) |\mc{P}| \stackeq{\Cref{obs:num-of-augpaths}}{\geq} (2 - 20\lambda)(\mu(G) - \mu(H)).
$$
Moving the terms and replacing $\lambda$ with $\epsilon/10$, we get $(3 - 2\epsilon)\mu(H) \geq (2-2\epsilon)\mu(G)$ and since $\epsilon \leq 1$,
$$
	\mu(H) \geq \frac{2 - 2\epsilon}{3 - 2\epsilon} \mu(G) \geq \frac{2 - 2\epsilon}{3} \mu(G) \geq \Big(\frac{2}{3} - \epsilon \Big)\mu(G).\qedhere
$$
\end{proof}

%% file: main.bbl
\newcommand{\etalchar}[1]{$^{#1}$}
\begin{thebibliography}{ABB{\etalchar{+}}19}

\bibitem[AB19]{AssadiBernsteinSOSA19}
Sepehr Assadi and Aaron Bernstein.
\newblock {Towards a Unified Theory of Sparsification for Matching Problems}.
\newblock In {\em 2nd Symposium on Simplicity in Algorithms, {SOSA} 2019,
  January 8-9, 2019, San Diego, CA, {USA}}, volume~69 of {\em {OASICS}}, pages
  11:1--11:20. Schloss Dagstuhl - Leibniz-Zentrum f{\"{u}}r Informatik, 2019.

\bibitem[AB21]{AssadiB21}
Sepehr Assadi and Soheil Behnezhad.
\newblock Beating two-thirds for random-order streaming matching.
\newblock In {\em 48th International Colloquium on Automata, Languages, and
  Programming, {ICALP} 2021, July 12-16, 2021, Glasgow, Scotland (Virtual
  Conference)}, volume 198 of {\em LIPIcs}, pages 19:1--19:13. Schloss Dagstuhl
  - Leibniz-Zentrum f{\"{u}}r Informatik, 2021.

\bibitem[ABB{\etalchar{+}}19]{AssadiBBMSSODA19}
Sepehr Assadi, MohammadHossein Bateni, Aaron Bernstein, Vahab~S. Mirrokni, and
  Cliff Stein.
\newblock {Coresets Meet {EDCS:} Algorithms for Matching and Vertex Cover on
  Massive Graphs}.
\newblock In {\em Proceedings of the Thirtieth Annual {ACM-SIAM} Symposium on
  Discrete Algorithms, {SODA} 2019, San Diego, California, USA, January 6-9,
  2019}, pages 1616--1635. {SIAM}, 2019.

\bibitem[Ber20]{BernsteinICALP20}
Aaron Bernstein.
\newblock {Improved Bounds for Matching in Random-Order Streams}.
\newblock In {\em 47th International Colloquium on Automata, Languages, and
  Programming, {ICALP} 2020, July 8-11, 2020, Saarbr{\"{u}}cken, Germany
  (Virtual Conference)}, volume 168 of {\em LIPIcs}, pages 12:1--12:13. Schloss
  Dagstuhl - Leibniz-Zentrum f{\"{u}}r Informatik, 2020.

\bibitem[BS15]{BernsteinSteinICALP15}
Aaron Bernstein and Cliff Stein.
\newblock {Fully Dynamic Matching in Bipartite Graphs}.
\newblock In {\em Automata, Languages, and Programming - 42nd International
  Colloquium, {ICALP} 2015, Kyoto, Japan, July 6-10, 2015, Proceedings, Part
  {I}}, volume 9134 of {\em Lecture Notes in Computer Science}, pages 167--179.
  Springer, 2015.

\bibitem[BS16]{BernsteinSteinSODA16}
Aaron Bernstein and Cliff Stein.
\newblock {Faster Fully Dynamic Matchings with Small Approximation Ratios}.
\newblock In {\em Proceedings of the Twenty-Seventh Annual {ACM-SIAM} Symposium
  on Discrete Algorithms, {SODA} 2016, Arlington, VA, USA, January 10-12,
  2016}, pages 692--711. {SIAM}, 2016.

\bibitem[Edm65]{edmonds1965paths}
Jack Edmonds.
\newblock Paths, trees, and flowers.
\newblock {\em Canadian Journal of mathematics}, 17:449--467, 1965.

\bibitem[GKK12]{DBLP:conf/soda/GoelKK12}
Ashish Goel, Michael Kapralov, and Sanjeev Khanna.
\newblock On the communication and streaming complexity of maximum bipartite
  matching.
\newblock In {\em Proceedings of the Twenty-Third Annual {ACM-SIAM} Symposium
  on Discrete Algorithms, {SODA} 2012, Kyoto, Japan, January 17-19, 2012},
  pages 468--485, 2012.

\bibitem[LP09]{MatchingTheoryBook}
L{\'a}szl{\'o} Lov{\'a}sz and Michael~D Plummer.
\newblock {\em Matching Theory}, volume 367.
\newblock American Mathematical Soc., 2009.

\end{thebibliography}
